\documentclass[10pt, utf8, UTF8]{wlscirep}
\usepackage[fontsize=8pt]{fontsize}
\usepackage[colaction]{multicol}
\usepackage{xcolor}
\usepackage{amsmath}
\usepackage{amsthm}
\usepackage{mathrsfs}
\usepackage{float}
\usepackage{enumitem}
\usepackage{caption}
\usepackage{hyperref}
\usepackage{lineno}
\usepackage{soul}

\newcommand{\f}{\frac}
\newcommand{\p}{\partial}
\newcommand{\B}{\mathbf}

\newtheorem{theorem}{Theorem}
\newtheorem{corollary}{Corollary}

\title{A Unified Data-Driven Framework for Efficient Scientific Discovery}
\author[1]{Tingxiong Xiao}
\author[1]{Xinxin Song}
\author[1]{Ziqian Wang}
\author[1]{Boyang Zhang}
\author[1,2,3,*]{Jinli Suo}

\affil[1]{Department of Automation, Tsinghua University, Beijing, China}
\affil[2]{Institute of Brain and Cognitive Sciences, Tsinghua University, Beijing, China}
\affil[3]{Shanghai Artificial Intelligence Laboratory, Shanghai, China}
\affil[*]{jlsuo@tsinghua.edu.cn}

\begin{abstract}
\textbf{
Scientific discovery drives progress across disciplines, from fundamental physics to industrial applications. However, identifying physical laws automatically from gathered datasets requires identifying the structure and parameters of the formula underlying the data, which involves navigating a vast search space and consuming substantial computational resources.
To address these issues, we build on the Buckingham $\Pi$ theorem and Taylor's theorem to create a unified representation of diverse formulas, which introduces latent variables to form a two-stage structure. To minimize the search space, we initially focus on determining the structure of the latent formula, including the relevant contributing inputs, the count of latent variables, and their interconnections. Following this, the process of parameter identification is expedited by enforcing dimensional constraints for physical relevance, favoring simplicity in the formulas, and employing strategic optimization techniques. Any overly complex outcomes are refined using symbolic regression for a compact form.
These general strategic techniques drastically reduce search iterations from hundreds of millions
to just tens, significantly enhancing the efficiency of data-driven formula discovery. We performed comprehensive validation to demonstrate FIND's effectiveness in discovering physical laws, dimensionless numbers, partial differential equations, and uniform critical system parameters across various fields, including astronomy, physics, chemistry, and electronics. The excellent performances across 11 distinct datasets position FIND as a powerful and versatile tool for advancing data-driven scientific discovery in multiple domains.
}
\end{abstract}
\begin{document}
\captionsetup[figure]{labelfont={bf},name={Fig.},labelsep=period}
\flushbottom
\maketitle
\thispagestyle{empty}

\section*{Introduction}
\begin{multicols}{2}  
Scientific discovery \cite{wang2023scientific,roscher2020explainable} forms the foundation of human understanding of the world and drives technological advancement, propelling developments from ancient astronomical observations\cite{ball2010data,kremer2017big,sen2022astronomical} to the modern theories of quantum mechanics\cite{han2021machine,rupp2015machine}. 
Traditional scientific discovery typically begins from formulating a theory, followed by experimental verification on collected data. 
Many classical laws were initially proposed based on theoretical derivations and subsequently confirmed through experimental observations \cite{born1962einstein}. 
However, with advancements in computer science, machine learning methods are now being utilized to support data-driven scientific discovery \cite{cornelio2023combining}. 
In this new paradigm, researchers apply machine learning algorithms to extract knowledge from a set of data \cite{lu2022discovering,rao2023encoding,tian2024machine,gao2024learning}, and then test the derived formulas to determine their accuracy \cite{chen2021physics,xie2022data,cory2024evolving}, or to provide insights that guide further exploration and research \cite{lauritsen2020explainable}.

Symbolic regression (SR) is a data-driven modeling approach that automatically uncovers the potential mathematical relationship between variables \cite{schmidt2009distilling,udrescu2020ai0,udrescu2020ai,tenachi2023deep,hernandez2019fast,kamienny2022end,cranmer2020discovering}. 
Typically, symbolic regression optimizes both the model structure and its parameters simultaneously, resulting in analytical formulas that describe complex functional relationships. 
The key implementation of symbolic regression is the searching algorithms, such as genetic programming \cite{koza1994genetic}, which iteratively evolve and optimize candidate solutions based on a fitness function. 
Due to its flexibility, independent of predefined structures, and high interpretability, symbolic regression is widely used in scientific research, engineering modeling, and data analysis \cite{purcell2023accelerating,liu2024automated,cory2024evolving,huang2023exploring,cranmer2023interpretable,reinbold2021robust}. 
However, symbolic regression suffers from high computational complexity and tends to overfit data \cite{wang2023scientific, udrescu2020ai0, udrescu2020ai, weng2020simple}.

Dimensional analysis (DA) is a mathematical method that focuses on the relationships between fundamental units of physical quantities, such as length, mass, and time \cite{tan2011dimensional,barenblatt2003scaling,brunton2016discovering,reynolds1883iii,ye2019energy}. It serves several purposes, including verifying the correctness of equations, simplifying complex problems, and deriving dimensionless parameters.
The foundation of dimensional analysis lies in the principle of dimensional homogeneity, which utilizes the Buckingham $\Pi$ theorem \cite{buckingham1914physically} to combine variables into dimensionless groups, ensuring that the resulting equations remain valid across different unit systems.
As a simple yet powerful tool, dimensional analysis is widely used in fields such as fluid mechanics, thermodynamics, and engineering sciences \cite{xie2022data,zhao2019bulk,gan2021universal,saha2021hierarchical,bakarji2022dimensionally,ma2024dimensional}. 
Despite its solid theoretical framework, identifying dimensionless groups requires extensive searching in a vast space, which hampers its practical application.

In this study, we propose a data-driven approach called FIND (Formulas IN Data) to tackle the high search costs and formula overfitting issues associated with the above data-driven discovery methods.
We first utilize the Buckingham $\Pi$ theorem and Taylor's theorem to create a unified representation of natural scientific laws, effectively transforming the problem of scientific discovery into an optimization task involving the structure and parameters of the target formula.
Specifically, our approach features an explicit two-stage structure: the $\Pi$ theorem guides the design of a latent layer to uncover hidden variables from the inputs, while Taylor's theorem shapes the expression layer, establishing an equivalent polynomial mapping from latent variables to output.
Next, we identify the structure of the target formula by identifying the contributing input variables using SHAP (Shapley Additive Explanations)\cite{lundberg2017unified} and the topological connections with the latent variables by designing a statistical inference scheme tailored for dense datasets. 
To further optimize the parameters efficiently, we impose dimensional constraints to ensure physically reasonable mathematical expressions as well as reduce the search space, and also introduce a progressive searching strategy to pursue the latent variables.
If the final expression is overly complex, symbolic regression is applied to simplify the polynomial into a more compact form.
The proposed approach functions as a universal framework that operates on diverse data with minimal manual configuration, and offers simple formulas with higher accuracy and significantly faster efficiency than existing methods.    

We performed thorough validation to showcase the versatility of FIND in four key areas: (i) recovering physical laws from astronomical observations, (ii) identifying dimensionless numbers in industrial systems, (iii) discovering partial differential equations in mechanical or electrical systems, and (iv) extracting the uniform critical parameters for circuit systems. 
Our results, based on 11 distinct datasets, establish FIND as an effective and general-purpose tool for data-driven scientific discovery across various domains.

\begin{figure*}[htbp]
\centering
\centerline{\includegraphics[width=\linewidth]{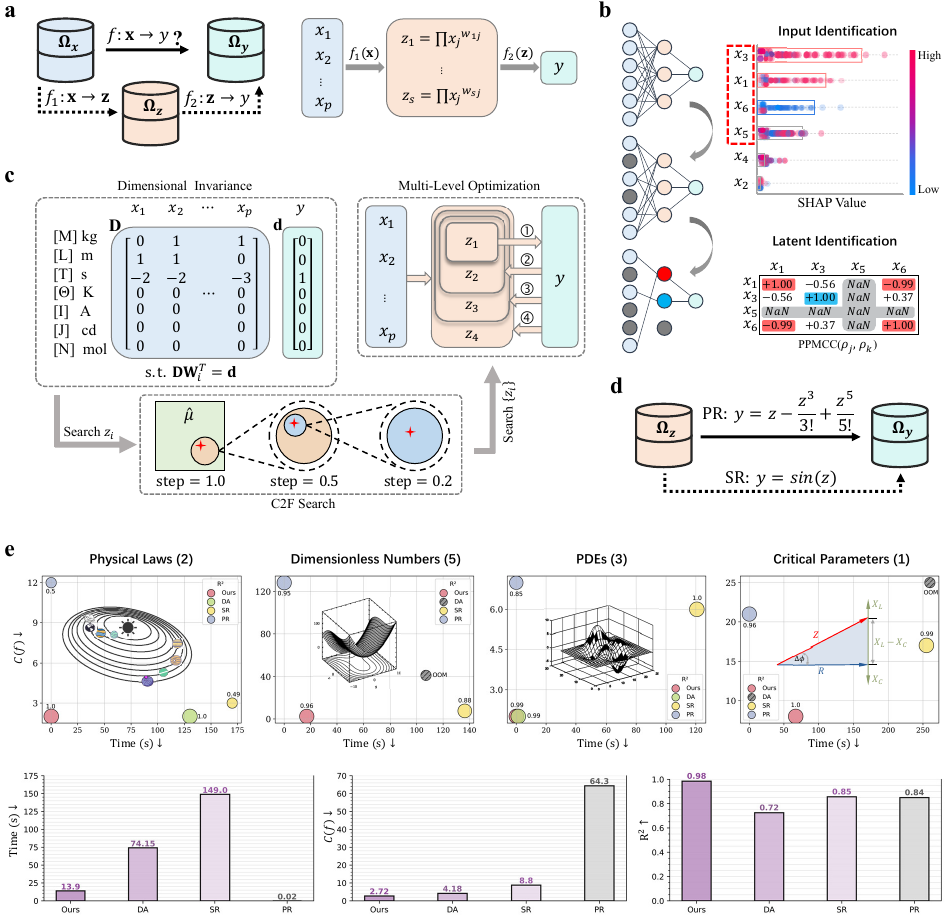}}
\caption{
\textbf{Overview of the proposed FIND framework.} 
The designed two-stage structure assumes that the dataset $(\Omega_x,\Omega_y)$ has an underlying low dimensional latent space $\Omega_z$, and decomposes the function $y=f(\B x)$ into $\B z=f_1(\B x)~\text{and}~y=f_2(\B z)$.
Inspired by the Buckingham $\Pi$ theorem, we define $\B z$ as a power product of $\B x$ to reduce its dimension; and according to Taylor's theorem, $f_2$ takes a polynomial form to achieve a simple but universal proxy to diverse formulas. 
\textbf{b,} Structure identification extracting the contributing inputs and latent variables. For inputs, we use SHAP to calculate and rank their contributions to the output, thereby identifying the most useful ones; for the latent variables, if the dataset is sufficiently dense for estimating partial derivatives, we can further calculate a PPMCC matrix and identify the key parameters of the latent variables (e.g., the number and the connections to the input variables). 
\textbf{c,} Parameter optimization of the function with identified structure. We employ dimensional invariance $\B D\B W_i^T = \B d$, with $\B D$ and $\B d$ being the dimensional matrices of the input and output, to ensure the correct physical meaning of expressions and reduce the search space. 
To search a single latent variable $z_i$, we propose an efficient grid searching algorithm that progressively refines the parameter from coarse to fine.
For multiple latent variables $\{z_i\}$, we introduce a multi-level optimization algorithm that optimizes the latent variables one by one until the performance of the discovered expression meets the requirements or stops improving. 
\textbf{d,} Polynomial and simplified symbolic description of the latent discovered formula. 
After determining the optimal latent variables, we obtain the polynomial expression of the formula $f_2$ via polynomial regression (PR), which can be further simplified via symbolic regression (SR) for a more compact form. 
\textbf{e,} Applications of FIND. We apply the proposed FIND for the discovery of the physical laws, dimensionless numbers, PDEs, and uniform critical parameters, covering a total of 11 datasets. 
The upper row compares the performance of FIND against dimensional analysis (DA), symbolic regression (SR), and polynomial regression (PR) on four tasks, in terms of the running efficiency of the algorithm (in runtime), the syntactic complexity ($C(f)$) and accuracy ($R^2$) of the discovered formulas. 
The charts in the lower row present the average performance of the above four methods on all 11 datasets.
}
\label{framework}
\end{figure*}

\vspace{2mm}
\section*{Results}

\subsection*{Overview of the Methods}
The framework of FIND is illustrated in Fig.~\ref{framework}. Given a dataset $(\Omega_x, \Omega_y)$, we assume that $\forall\B x\in\Omega_x$ there exists a function $y=f(\B x) \in \Omega_y$, and it can be decomposed into two successive mappings to a latent space $\B z=f_1(\B x)\in \Omega_z$ first and then to the output $y=f_2(\B z)$, as shown in Fig.~\ref{framework}a. Denoting $\B x = <x_j> \in \mathbb{R}^p$, $\B z = <z_i> \in \mathbb{R}^s$, $\B W = <w_{ij}> \in \mathbb{R}^{s\times p}$ and drawing inspiration from Buckingham $\Pi$ theorem\cite{buckingham1914physically} and Taylor's theorem, we set
\begin{eqnarray}
\begin{split}
    f_1:& ~~z_i=\prod_{j=1}^{p}x_j^{w_{ij}},    \\
    f_2:&~~y=a_0+\sum_{i=1}^s a_i z_i+\sum_{i_1,i_2=1}^s a_{i_1i_2} z_{i_1} z_{i_2}+\ldots.
\end{split}
\label{latent0}
\end{eqnarray}

For efficient discovery of the latent mapping function, we first identify the structure of the two-stage mapping to reduce $f$'s search space, as shown in Fig.~\ref{framework}b. 
First, we calculate and rank the contribution of each input variable to the output using SHAP\cite{lundberg2017unified} and determine the most contributing input variables. 
Furthermore, for dense datasets that allow estimating first-order partial derivatives, we compute the Pearson product-moment correlation coefficient (PPMCC) between $\rho_j=x_j\f{\p y}{\p x_j}$ and $\rho_k=x_k\f{\p y}{\p x_k}$. The PPMCC matrix then subsequently reveals the structural information of the mapping $f$, including the contributing inputs, the number of latent variables, and the connections between them. 

After determining the structure of $f$, we need to optimize the function parameters, i.e., $\B W$ in $f_1$ and $\{a_i\}$ in $f_2$. In Fig.~\ref{framework}c, we analyze the units of $\B x$, $\B z$, and $y$ to ensure dimensional invariance and further reduce the search space. 
For each latent variable $z_i$, we adopt a coarse-to-fine (C2F) grid search on $\B W_i$ to obtain the optimal solution gradually. For multiple latent variables $\B z\in\mathbb{R}^s$, we introduce a multi-level optimization algorithm that learns the set of latent variables one by one, until the performance of the
discovered expression meets the requirements. 
Regarding $\{a_i\}$, polynomial regression is performed directly on $\B z$ and $y$ using the least squares method.
So far, we have achieved the optimal latent variables and their polynomial expression combination defined in $f_2$, and proceed to further simplify the expression. Specifically, we use PySR\cite{cranmer2023interpretable} to perform symbolic regression on the data points in $\Omega_z$ and $\Omega_y$, achieving a more compact representation, as illustrated in Fig.~\ref{framework}d.

The proposed discovery approach is general and applicable to various fields. 
We demonstrate FIND's versatility across 11 datasets covering tasks such as physical laws discovery (2), dimensionless numbers discovery (5), PDEs discovery (3), and critical parameters identification (1). In Fig.~\ref{framework}e, we demonstrate the advantageous performance of FIND over existing methods on four tasks, in terms of running efficiency, simplicity of the discovered formula, and the data-fitting accuracy. The upper row shows separate performance comparisons on each dataset, and the graphs in the lower row compare the average performance of all 11 datasets.
The results demonstrate that FIND can efficiently discover formulas with low complexity, high performance metrics, and physical interpretability. 


\begin{figure*}[htbp]
\centering
\centering
\centerline{\includegraphics[width=0.87\linewidth]{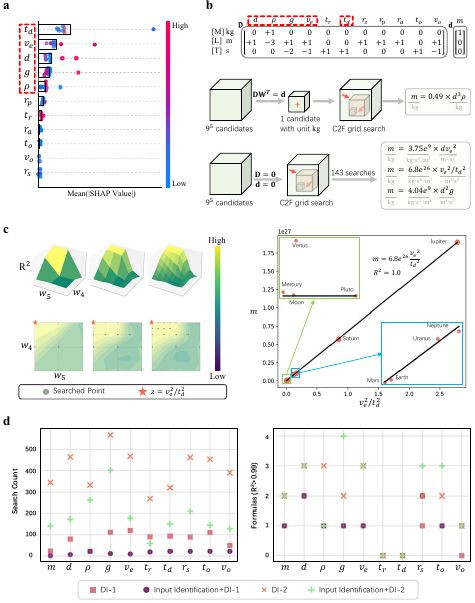}}
\caption{
\textbf{Illustration of FIND for physical laws discovery, with the solar system planets dataset as an example.} \textbf{a,} Input identification using SHAP, which quantifies the average contribution of each input variable to the output, and the 5 most influential variables were selected as the final input of the final formula. 
\textbf{b,} Dimensional invariance (DI) constraints. top: $\B D$ and $\B d$ are the dimension matrices of the input and the output. Here limit $\B W$ to $[-2, +2]^5$ and the minimum precision 0.5, and $9^5$ searches are required. 
middle: 
The strategy imposing DI constraint during parameter optimization (DI-1) by enforcing $\B D\B W=\B d$, which directly reduces the $9^5$ candidate values to just one and ultimately obtains the result $m=0.49d^3\rho$. 
If multiple candidate values satisfy $\B D\B W=\B d$, we further perform a C2F search in the small candidate set.  
bottom:
The strategy applying unit constraints after C2F parameter optimization, which reduces the initial $9^5$ candidates to 143. 
Ultimately, we obtain three high-scoring formulas and assign proper units to the constants to ensure dimensional consistency between the two sides of the formulas.
\textbf{c,} left: A visualization of the C2F grid search of $v_e$'s power $w_4$ and $t_d$'s power $w_5$, with the searched optimal values being 2 and -2, i.e., $z=v_e^2/t_d^2$.  right: the discovered polynomial function with high score ($R^2 = 1.0$) based on the identified latent variable $z$ in the left panel. 
\textbf{d,} Performance comparison under four different settings (two DI strategies, and whether adopting input identification), when selecting different variables as the output. Here we systematically compare the number of required search iterations (left) and the number of discovered high-$R^2$ formulas (right).
}
\label{nasa}
\end{figure*}

\subsection*{Physical Laws Discovery}
To showcase FIND's capability and efficiency in uncovering physical laws, we rediscover the astronomical formulas from actual measurements of the planets in the solar system (reported in the NASA Planetary Fact Sheet\cite{nasa2017factsheet}) and exoplanet recordings (from the NASA exoplanet archive\cite{nasa2017archive}).

For planets in the solar system, the set of variables $S=\{m,d_0,\rho_0,g,v_e,t_r,t_d,r_s,r_p,r_a,t_o,v_o\}$ includes the mass of the planet $m$, diameter $d$, density $\rho$, gravitational acceleration $g$, escape velocity $v_e$, rotation period $t_r$, length of day $t_d$, distance from the sun $r_s$, perihelion $r_p$, aphelion $r_a$, orbital period $t_o$, orbital velocity $v_o$, with standard units. The objective is to find $y=f_2\circ f_1(S\backslash\{y\}),~\forall y\in S$.

Taking the planetary mass $m$ as an example, after identifying the contributing variables--- $t_d,v_e,d,g,\rho$  following the steps in Fig.~\ref{nasa}a, the basic physical quantities involved here include mass (kg), length (m), and time (s). Targeting for simple formulas, we set $s=1$ to search for a single latent variable $z$ and the expressions predefined in Eq.~\ref{latent0} turn into
\begin{equation}
\begin{split} 
&f_1:z=d^{w_1}\rho^{w_2}g^{w_3}v_e^{w_4}t_d^{w_5},\\
    &f_2:m=\sum_{k=0}^n a_kz^k.
\end{split}
\end{equation}
We denote the target weights as $\B W=[w_1,\ldots,w_5]$, the dimensional matrices of the contributing input variables and the output $m$ respectively as $\B D$ and $\B d$, and have 
\begin{equation}
\B D=\left[\!
    \begin{array}{rrrrr}
        0,& +1,& 0,& 0,& 0 \\
        +1,& -3,& +1,& +1,& 0 \\
        0,& 0,& -2,& -1,& +1
    \end{array}
\!\right], \\
\B d=\left[\!
    \begin{array}{c}
        1 \\
        0 \\
        0
    \end{array}
\!\right].
\end{equation}
To force the unit of the latent variable to be consistent with the output, we have the following linear nonhomogeneous equation.
\begin{equation}
    \B D\B W^T=\B d.
\label{unit constrain}
\end{equation}
Using the Gaussian elimination method, one can calculate the basis and bias vectors of this equation as
\begin{equation}
\small
\B E_1=\left[\!
    \begin{array}{r}
        -0.5 \\
        0.0 \\
        -0.5 \\
        +1.0 \\
        0.0
    \end{array}
\!\right], \\
\B E_2=\left[\!
    \begin{array}{r}
        -0.5 \\
        0.0 \\
        +0.5 \\
        0.0 \\
        +1.0
    \end{array}
\!\right], \\
\B e^*=\left[\!
    \begin{array}{r}
        3 \\
        1 \\
        0 \\
        0 \\
        0
    \end{array}
\!\right],
\end{equation}
and derive its closed-form solution:
\begin{equation}
    \B W^T=\lambda_1\B E_1+\lambda_2\B E_2+\B e^*=\B E\boldsymbol\lambda^T+\B e^*.
\end{equation}
So far, the task of optimizing $\B W\in\mathbb{R}^{1\times 5}$ turns into searching for $\boldsymbol\lambda\in\mathbb{R}^{1\times 2}$ that fits the data observation with high accuracy. 

As shown in Fig.~\ref{nasa}b, here we have established two schemes for imposing dimensional invariance: DI-1 that applies unit constraints during search, and DI-2 that applies post-search unit constraints.
The final discovered formula of the former scheme is $m=0.49d^3\rho$, which is consistent with the theoretical result $m=\pi d^3\rho/6\approx 0.52d^3\rho$.  The latter obtained 3 formulas: $m=3.75e^9dv_e^2$, $m=6.8e^{26}v_e^2/t_d^2$, and $m=4.04e^9d^2g$. The first and third formulas are both consistent with the theoretical values, while the second one remains to be explored. We conducted further investigation into this unexplored formula and found that it primarily holds for distant massive planets, whereas smaller-mass planets deviate to some extent, especially Mercury and Venus. Such unexplored expression indicates the potential of FIND in discovering new formulas from data recordings, thereby inspiring worthwhile astronomical studies. 
In both schemes, we conduct efficient parameter search in a coarse-to-fine manner. For easy demonstration, here we only show the optimization process of $w_4$ and $w_5$ (the most influential inputs to $m$, $v_e$, and $t_d$) in Fig.~\ref{nasa}c. One can observe that the optimal estimation of $(w_4, w_5)$ is $(+2, -2)$, with an $R^2$ score as high as 0.9999 and the resulting functions being
\begin{equation}
\begin{split}
    &f_1: z=v_e^2/t_d^2, \\
    &f_2: m=-1.38e^{24}+6.8e^{26}z\approx 6.8e^{26}z.
\label{formula}
\end{split}
\end{equation}

We further treat variables other than $m$ as outputs, and follow the same procedure to discover additional formulas.  
In sum, we have found a total of 58 formulas, among which only the formula in Eq.~\ref{formula} lacks theoretical proof, while all others can be theoretically derived. The 58 formulas, along with their theoretical proofs and the validation of Eq.~\ref{formula}, are detailed in the Supplementary Materials.

The left panel in Fig.~\ref{nasa}d demonstrates that pinpointing contributing inputs and imposing dimensional invariance significantly reduce the number of search iterations. For example, during the search for the formulas that output $d$, DI-2 needed 464 iterations and DI-1 required 80 iterations; after introducing input identification, these numbers drop significantly to 173 and 6 iterations, respectively.
Comparing two schemes of applying dimensional invariance constraints, DI-2 has an advantage in that it is more flexible and can discover a wider range of potential formulas. The right panel in Fig.~\ref{nasa}d shows that DI-2 can discover more valid formulas (with $R^2>0.99$). Notably, when combined with input identification, it not only reduces the search space but also yields an increased number of identified formulas.


We have also validated the effectiveness of FIND in discovering physical formulas using an exoplanet dataset. In particular, we gathered and integrated data from Solar, Trappist-1, KOI-351, and GJ 667 based on the latest research publications, and then used FIND to explore the relationship between planetary orbital periods and other variables. 
The optimal solution was obtained in just 26 search iterations within 0.14 seconds, and the final solution $t_o=5.43e^{-10}r_o^{1.5}$ verifies Kepler's third law. For details, please refer to the Supplementary Materials.

\subsection*{Dimensionless Numbers Discovery}

\begin{figure*}[htbp]
\begin{center}
\centerline{\includegraphics[width=0.98\linewidth]{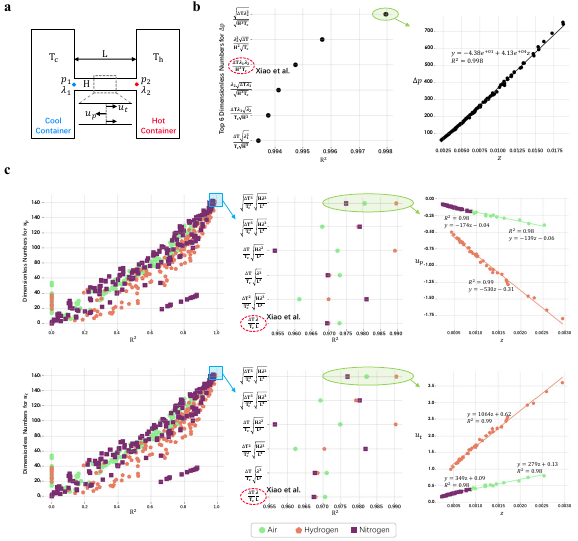}}
\caption{\textbf{Demonstration of applying FIND for dimensionless numbers discovery.} \textbf{a,} An illustration of a typical hydrogen Knudsen compressor, consisting of a cold container, a hot container, and a microchannel (with length $L$ and height $H$). 
When there exists a temperature difference $\Delta T=T_h-T_c$ between containers, a pressure difference $\Delta p=p_2-p_1$ occurs across the microchannel with the mean free paths at the two sides being $\lambda_1$ and $\lambda_2$, accompanied by liquid flow. 
The Poiseuille flow $u_p$ mainly distributes near the centerline of the microchannel, while the thermal transpiration flow $u_t$ nears the wall. 
\textbf{b,} The top 6 dimensionless numbers for predicting $\Delta p$ (left) and the corresponding prediction function of the highest-ranked dimensionless number $z=\sqrt{\f{\Delta T\lambda_2^3}{H^3T_c}}$, which exhibits a linear mapping $\Delta p=-4.38e^{+01}+4.13e^{+04}z$ (right). The dimensionless number $\f{\Delta T\lambda_1\lambda_2}{H^2T_c}$ (highlighted with red ellipse) was previously identified in \cite{xiao2023dimensional}. 
\textbf{c,} Discovered dimensionless numbers for Poiseuille flow $u_p$ (upper) and thermal transpiration flow $u_t$ (lower), by applying FIND on air, hydrogen, and nitrogen measurements separately. 
For each prediction variable ($u_p$ or $u_t$), the left panel presents the discovered 163 dimensionless numbers shared by three gases, which show highly-consistent trends. We show the zoomed-in view of the top 6 dimensionless numbers with high $R^2$ scores in the middle panel, including a previously identified dimensionless number $\f{\Delta T\lambda}{T_cL}$\cite{xiao2023dimensional}. In the right panel, we further plot the prediction functions of the highest-ranked dimensionless number $z=\sqrt{\f{\Delta T^3}{T_c^3}}\sqrt{\f{H\lambda^2}{L^3}}$. 
Here, three gases share the same dimensionless number but exhibit three distinct prediction functions.}
\label{exp2}
\end{center}
\end{figure*}

We demonstrate FIND's applications in discovering dimensionless numbers from data recordings in 5 different industrial or chemistry scenarios.   

We first present the results on the working status of a hydrogen compressor as an example, using a dataset from the work of Xiao et al.\cite{xiao2023dimensional} which contains 120 measurements, with 40 for air, 40 for hydrogen, and 40 for nitrogen. 
A hydrogen Knudsen compressor typically consists of a cold container, a hot container, and a microchannel, as illustrated in Fig.~\ref{exp2}a. 

We first set the microchannel length $L$, microchannel height $H$, cool container's temperature $T_c$, temperature difference between two containers $\Delta T=T_h-T_c$, the mean free paths of the measurement points $\lambda_1$, $\lambda_2$ as input variables, and try to find a dimensionless number that can predict the pressure difference $\Delta p=p_2-p_1$. 
Since the measurements are sparse, we cannot calculate the partial derivatives reliably to determine the number of latent variables. Instead, we start from a single latent variable and increase the number progressively until meeting the desired data fitting accuracy (e.g., $R^2>0.95$). Even with a single variable, i.e., $s=1$, we found 10 dimensionless numbers with $R^2>0.99$ after 681 searches in 1.13s, and the top 6 are shown in Fig.~\ref{exp2}b. 
Here $\f{\Delta T\lambda_1\lambda_2}{H^2T_c}$ is a discovered dimensionless number in \cite{xiao2023dimensional}, and the rest are all new ones among which 2 are with higher $R^2$. $z=\sqrt{\f{\Delta T\lambda_2^3}{H^3T_c}}$ exhibits the highest $R^2$ with the functional relationship $\Delta p=-4.38e^{+01}+4.13e^{+04}z$.
Our findings indicate that the pressure difference $\Delta p$ correlates strongly with dimensionless combinations ($\Delta T/T_c, \lambda_1/H, \lambda_2/H$) and exhibits no dependence on $L$.

Secondly, we tried to find dimensionless numbers for Poiseuille flow $u_p$ distributed mainly near the centerline of the microchannel, and the thermal transpiration flow $u_t$ that mainly near the wall. Here we set $L$, $H$, $T_c$, $\Delta T$ and the mean free paths of the microchannel center $\lambda$ as input variables for prediction. 
In parameter search, when $s$ increases from 1 to 2, $R^2$ only sees a slight improvement from 0.8889 to 0.8944, suggesting that the number of latent variables is likely not the issue.
We then divided 120 data points into air, hydrogen, and nitrogen based on gas types and conducted separate experiments. In Fig.~\ref{exp2}c, the top-left subplot presents the search results for $u_p$, where a total of 163 dimensionless numbers were examined, showing consistent trends across all three gases. The top 6 dimensionless numbers along with their $R^2$ metrics are displayed in an expanded format, including a previously identified dimensionless number $\f{\Delta T\lambda}{T_cL}$\cite{xiao2023dimensional}. 
The dimensionless number $\sqrt{\f{\Delta T^3}{T_c^3}}\sqrt{\f{H\lambda^2}{L^3}}$ emerges as the optimal scaling parameter to characterize $u_p$, with its functional relationships for three distinct gas types presented in the top right subplot. We observed the same dimensionless number appearing across different gases, though with varying expressions. We can also see that the discovered dimensionless numbers contain a large number of forms such as $\Delta T/T_c$, $\lambda/L$, and $H/L$, which have specific physical meanings.

We also attempted to discover dimensionless numbers in four other scenarios: keyholes in laser–metal interaction\cite{zhao2019bulk, gan2021universal}, porosity formation in 3D printing of metals\cite{wang2019dimensionless,kasperovich2016correlation,kumar2019influence,cherry2015investigation,leicht2020effect,simmons2020influence}, turbulent Rayleigh-Bénard convection\cite{xie2022data}, and electrocatalytic reactions\cite{lin2024machine}. We identified some classic dimensionless numbers such as the Keyhole number\cite{gan2021universal, ye2019energy} and Rayleigh number, some previously discovered numbers \cite{xie2022data,lin2024machine},  as well as many new ones with higher metrics. The results demonstrate that FIND can successfully rediscover known dimensionless numbers as well as identify new ones with higher metrics. Detailed experimental procedures and results are provided in the Supplementary Materials. 

\subsection*{PDEs Discovery}

\begin{figure*}[ht]
\begin{center}
\centerline{\includegraphics[width=0.98\linewidth]{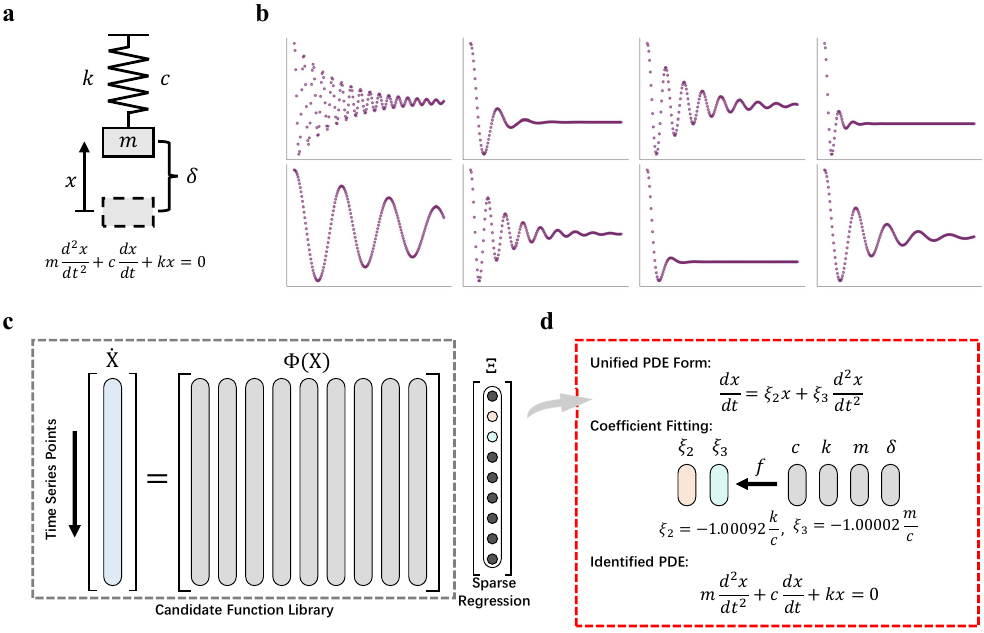}}
\caption{\textbf{Demonstration of FIND's capability of PDEs discovery.} \textbf{a,} The structure of the spring-mass-damper system, consisting of a mass ($m$), a spring ($k$), and a damper ($c$), with initial position $x(0)=\delta$. The system is a classic example of a second-order mechanical model whose PDE equation is $m\f{d^2x}{dt^2}+c\f{dx}{dt}+kx=0$. 
\textbf{b,} 
Eight exemplar time series generated by varying parameters $c, k, m, \delta$, with time ranges in $0\sim 20$s, and each time series uniformly sampled to obtain over 800 sampling points. 
\textbf{c,} The matrix equation for identifying the governing equation of each system using SINDy, i.e., 
$\dot{\B X}=\Phi(\B X)\Xi$, in which $\dot{\B X}$ is the vector composed of $\f{dx}{dt}$ values at different data points, the basis set $\Phi(\B X) = \{1,x,\f{d^2x}{dt^2},x^2,x\f{dx}{dt},x\f{d^2x}{dt^2},(\f{dx}{dt})^2,\f{dx}{dt}\f{d^2x}{dt^2},(\f{d^2x}{dt^2})^2\}$, 
and the coefficient vector $\Xi=\{\xi_1,\xi_2,\cdots,\xi_9\}$ is ultimately determined via sparse regression on the observed data. The PDEs for different systems follow a unified form $\f{dx}{dt}=\xi_2x+\xi_3\f{d^2x}{dt^2}$, but with varying coefficients $\xi_2$ and $\xi_3$. 
\textbf{d,} The unified PDE discovered by FIND from 8 time series datasets of different system parameters $c, k, m, \delta$, obtained by establishing the coefficients between $\xi_2,\xi_3$ and $c, k, m, \delta$ and ultimately identifying the unified PDE of the spring-mass-damper system.}
\label{exp3}
\end{center}
\end{figure*}

In this section, we evaluate FIND's ability to uncover latent PDEs using data from three different fields in mechanics or electronics.

We first tested on the spring-mass-damper system, which
consists of a mass $m$ characterized by its position $x(t)$, with $x(0)=\delta$, a spring with stiffness $k$, and a damper with coefficient $c$, as shown in Fig.~\ref{exp3}a. It is a classic example of a second-order mechanical model, and the underlying PDE equation is
\begin{equation}
\begin{split}
    m\f{d^2x}{dt^2}+c\f{dx}{dt}+kx=0.
\end{split}
\label{rlc pde}
\end{equation}

By changing the parameters $c, k, m, \delta$, we obtained multiple sets of time series data spanning 20s, as shown in Fig.~\ref{exp3}b. 
Next, we apply SINDy\cite{brunton2016discovering} to identify the governing equations from each time series dataset, as illustrated in Fig.~\ref{exp3}c. For each time series, we uniformly sample 800 points over the interval $t\in[0,20]$, estimate their derivatives through polynomial fitting and differentiation between neighboring points to construct the regression library. Afterward, we fit the following equation via sparse regression
\begin{equation}
    \dot{\B X}=\Phi(\B X)\Xi,
\end{equation}
where $\dot{\B X}$ contains the $\f{dx}{dt}$ values at each data point, $\Phi(\B X)$ contains 9 terms, and $\Xi\in\mathbb{R}^9$ contains the target coefficients. The results show that for each time series dataset, only coefficients $\xi_2$ and $\xi_3$ were non-zero, while all other coefficients vanished, leading to a unified form for all the governing equations
\begin{equation}
    \f{dx}{dt}=\xi_2x+\xi_3\f{d^2x}{dt^2}.
\end{equation}

We observed that $\xi_2$ and $\xi_3$ vary with the system parameters $c, k, m, \delta$.  
Based on the dimensional consistency of the PDE, the units of $\xi_2$ and $\xi_3$ are derived as $s^{-1}$ and $s$, respectively. Using FIND to fit these relationships, we obtained a concise relationship $\xi_2=-\f{k}{c},~\xi_3=-\f{m}{c}$, which exactly recovers the formula given in Eq.~\ref{rlc pde}.

Traditional PDE systems typically require specialized physics knowledge and are derived through theoretical analysis. Here, we propose a novel data-driven approach for PDE discovery that operates entirely without prior expertise. By fixing the system structure while varying its parameters, and combining SINDy with FIND, we successfully recovered the true governing PDE of the system. We believe this approach will enable non-experts to identify the PDE form of systems or achieve the identification of black-box systems with minimal effort.

We further validated the effectiveness of our method for PDEs discovery in two additional scenarios: identification of the Karman vortex street problem\cite{xie2022data} and discovery of the governing equation for a passive series RLC circuit. FIND successfully extracted the Navier-Stokes equation and the second-order governing equation of the RLC circuit from two sparse time-series datasets, respectively, demonstrating high consistency with theoretical results. This further validates the effectiveness of its PDE discovery capability. Detailed experimental procedures and results are provided in the Supplementary Materials.

\subsection*{Critical Parameters Identification}

\begin{figure*}[ht]
\begin{center}
\centerline{\includegraphics[width=\linewidth]{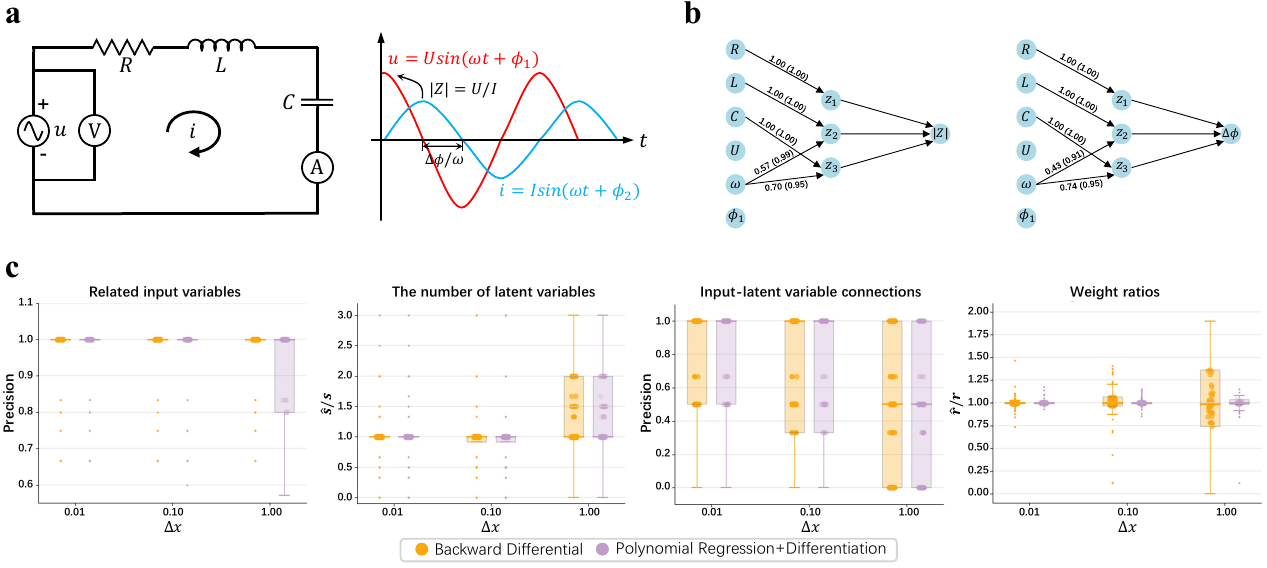}}
\caption{\textbf{Demonstration of using FIND for discovering the critical system parameters of a series RLC AC circuit, from dense observations.} \textbf{a,} The structure of a series RLC AC circuits, consisting of a resistor $R$, an inductor $L$, and a capacitor $C$, with an input voltage $u=Usin(\omega t+\phi_1)$ and output current $i=Isin(\omega t+\phi_2)$. 
We use $R$, $L$, $C$, $U$, $\omega$, and $\phi_1$ as independent variables to predict the two key parameters of the circuit---impedance magnitude of the circuit $|Z|=U/I$ and the phase difference $\Delta\phi=\phi_1-\phi_2$. 
\textbf{b,} The results of the identified latent structures of the two target prediction values. 
The numbers on each arrow represent the estimated weight ratio and the corresponding confidence (in round brackets). 
\textbf{c,} Accuracy of the identified latent structure at three data densities, using two strategies for partial derivative estimation---backward differential and differentiation following polynomial fitting with neighboring points. 
From left to right, we plot the accuracy of the identified contributing inputs, the number of latent variables, the connections between inputs and the latent variables, and the estimations of connection weights, with the theoretical optimal accuracies being all 1.0. 
}
\label{exp4}
\end{center}
\end{figure*}



In system analysis, identifying critical parameters is essential for understanding and optimizing performance. These parameters often determine a system's dynamic response, stability, and interactions with other systems. However, some systems involve multiple latent variables, which may lead to the curse of dimensionality. 


In the above experiments, we pursue the latent variables progressively to meet the accuracy criteria, thereby dramatically reducing the search space.
Here, for datasets with sufficiently dense observations supporting partial differential estimation, we propose a novel algorithm that identifies the topological structure of the two-stage function (including the contributing inputs, the number of latent variables, as well as the connection between) directly from the data observations. An explicit topological structure can significantly narrow down the search space. 

Taking the series RLC AC circuit in Fig.~\ref{exp4}a as an example, the circuit consists of a resistor $R$, an inductor $L$, and a capacitor $C$, with an input voltage $u=Usin(\omega t+\phi_1)$ and output current $i=Isin(\omega t+\phi_2)$. By varying $R$, $L$, $C$, $U$, $\omega$, $I$, and $\phi_1$, we obtained 31,250 pairs of $u$ and $i$ measurements, and used these variables to predict the two critical parameters of the circuit---impedance magnitude of the circuit $|Z|=U/I$ and the phase difference $\Delta\phi=\phi_1-\phi_2$.



We set $|Z|$ as the output and first calculated $\boldsymbol\rho=<\rho_j> \in \mathbb{R}^6$ at each sampling point, where $\rho_j=x_j\f{\p y}{\p x_j}$. Then, we computed the Pearson product-moment correlation coefficient for $\rho_j$ and $\rho_k$ and obtained a PPMCC matrix
\begin{equation}
\small
\begin{split}
\text{PPMCC}_{|Z|}=
\left[\!
    \begin{array}{rrrrrr}
        \bold{+1.00}, & -0.37, & -0.16, & NaN, & -0.32, & NaN \\
        -0.37, & \bold{+1.00}, & +0.90, & NaN, & \bold{+0.99}, & NaN \\
        -0.16, & +0.90, & \bold{+1.00}, & NaN, & \bold{+0.95}, & NaN \\
        NaN, & NaN, & NaN, & NaN, & NaN, & NaN \\
        -0.32, & \bold{+0.99}, & \bold{+0.95}, & NaN, & \bold{+1.00}, & NaN \\
        NaN, & NaN, & NaN, & NaN, & NaN, & NaN 
    \end{array}
    \!\right].
\end{split}
\end{equation}
If $|\text{PPMCC}(\rho_j, \rho_k)|\geq 0.95$, we consider that $x_j$ and $x_k$ connect with the same latent variable. 
We can observe that there are three latent variables: $z_1$ is only associated with $R$; $z_2$ is related to both $L$ and $\omega$ with $sign(w_{22})=sign(w_{25})$; $z_3$ is connected to $C$ and $\omega$ with $sign(w_{33})=sign(w_{35})$; $U$ and $\phi_1$ are un-correlated variables, as illustrated in the left panel of Fig.~\ref{exp4}b. Similarly, the latent structure identification process for $\Delta\phi$ follows the same methodology, with the final results illustrated in Fig.~\ref{exp4}b, right.

The identified structure can serve as a constraint for the search space of $\B W$, which reduces the $9^{3\times 6}$ trials required before latent structure identification to only 342 ones. The search takes 19 seconds in total, and the final results are 
\begin{equation}
\begin{split}
    &|Z|=f_{|Z|}(\f{1}{R},\omega^2L^2,\f{1}{\omega C}), \\
    &\Delta\phi=f_{\Delta\phi}(\f{1}{R^2},\omega L,\f{1}{\omega^2C^2}).
\end{split}
\end{equation}
The obtained functions $f_{|Z|}$ and $f_{\Delta\phi}$ are both 5th-order polynomials, which are highly complex and are not intuitive for human understanding. After determining the optimal latent variables, we used PySR\cite{cranmer2023interpretable} to perform symbolic regression on the latent variables and the outputs, simplifying $f_{|Z|}$ and $f_{\Delta\phi}$ to the following equations:
\begin{equation}
\small
\begin{split}
    &|Z|=\sqrt{z_1^{-2}+(\sqrt{z_2}-z_3)^2}=\sqrt{R^2+(\omega L-\f{1}{\omega C})^2}, \\
    &\Delta\phi=arctan(\sqrt{z_1}(z_2-\sqrt{z_3}))=arctan(\f{\omega L-\f{1}{\omega C}}{R}),
\end{split}
\end{equation}
which aligns well with the theoretical formulas. 

Note that latent structure identification requires estimating partial derivatives, so the data density and derivative calculator are both important. Here, we present the performance across three datasets with different sparsity levels, using two derivative estimation schemes. 
We generated 100 composite functions $y=f_2(f_1(\B x))$ with $p\in [3,8]$ and $s\in [1,min(p,4)]$, where $f_1$ had random connections/weights between $\B x$ and $\B z$, and $f_2$ was constructed through random combinations of mathematical operators. We sampled each function at three intervals ($\Delta x=0.01, 0.10, 1.00$), producing three datasets with identical functions but different sparsity levels.
The identified latent structures are evaluated in terms of the accuracy of identifying i) the contributing variables, ii) the number of latent variables, iii) the connection between inputs and the latent variables, as well as iv) the weight of the connection. The accuracy of i) and iii) is assessed using $Precision = TP / (TP + FN)$, ii) and iv) are evaluated by the ratio between the estimated version and the true values. The ideal scores for these four evaluations are all 1.0, and the results are shown in Fig.~\ref{exp4}c. 
It can be observed that a smaller $\Delta x$ (i.e., denser data) can produce more accurate derivatives, and thus a more accurate latent structure. The backward difference method provides better identification of contributing input variables, while the differentiation following polynomial fitting yields more accurate weight estimation. 

Our proposed identification algorithm for dense datasets successfully resolves the two-stage mapping of the critical system parameters, overcoming the curse of dimensionality from multiple latent variables, and can be applied in controlled variable experiments in the scientific field.

\section*{Discussions}

In this article, we introduce a new data-driven approach for discovering implicit formulas from observations across various fields, called FIND (Formulas IN Data). Our method utilizes the Buckingham $\Pi$ theorem and Taylor's theorem to unify the representation of scientific laws into an explicit two-stage composite function. This allows us to express scientific discoveries as an optimization involving the structure and parameters of the latent function.
The optimization demands navigating a vast solution space, and we proposed to strategically cut down the computational cost from hundreds of millions to just tens, by identifying the topological structure of the two-stage function, as well as accelerating parameter optimization by forcing dimensional invariance, conducting a C2F grid search, and pursuing latent variables progressively with a multi-level optimization algorithm. 
The resulting polynomials act as temporary proxies for the original function, facilitating the rapid identification of the optimal latent variable. If the identified polynomial is overly complex or if a more physically meaningful formulation is desired, symbolic regression can be employed to refine it into a more concise and compact expression.

Compared to traditional scientific discovery methods such as dimensional analysis (DA), symbolic regression (SR), and polynomial regression (PR), the FIND method offers notable advantages in speed, formula simplicity, data-fitting accuracy, and physical interpretability. However, FIND does not completely disregard these existing methods; instead, it draws insights from all three approaches. By incorporating well-designed mechanisms, FIND achieves performance that surpasses existing methods and demonstrates wide applicability across various domains.

We demonstrate the versatility of FIND across 11 datasets, encompassing tasks such as the discovery of physical laws, dimensionless numbers, partial differential equations, and critical system parameters. 
In the discovery of physical laws, we identified 58 high-scoring formulas from the planetary dataset of the solar system, 57 of which could be theoretically validated, and the remaining one experimentally validated as applicable to large-mass planets and inspires further astronomical studies of the underlying laws or working mechanism. 
Regarding the discovery of dimensionless numbers, our approach not only rediscovers known dimensionless numbers but also uncovers many new ones with higher scores. These new numbers can be crucial for understanding, predicting, and controlling system outputs.
In terms of PDEs discovery, the combination of SINDy and FIND enables the direct derivation of the PDE of a system from observed data, eliminating the need for domain expertise and opening up possibilities for applications in system identification.
For the data with sufficient density to calculate the partial derivatives, our algorithm efficiently infers the number of latent variables alongside their connections with the contributing inputs, which helps to identify the critical system parameters.

Despite the high performance and wide applicability, the FIND framework still has some limitations that need to be addressed. Firstly, while many functions can be approximated locally by polynomials, we lack prior knowledge about the optimal polynomial degree. As a result, we typically use a 5th-order polynomial by default, which is sufficient for most cases. 
We now offer customized settings to the users, and an extension with flexible estimation of the polynomial orders can be another option.
Secondly, the algorithm requires estimating partial derivatives to explicitly identify the latent structure of the target formula, making it most suitable for dense datasets, particularly those where variables can be controlled experimentally. We are open to more advanced mathematical techniques for estimating partial derivatives from sparse observations. 
Another limitation is that the latent variables in FIND are expressed solely in the form of power products, which can be addressed by extending the candidate terms of the input to cover the formulas beyond the approximation capability of polynomials.

\newpage
\section*{Methods}


\subsection*{Buckingham $\Pi$ Theorem and Taylor’s Theorem}
\begin{theorem}
{\rm(Buckingham $\Pi$ Theorem\cite{buckingham1914physically})} Assuming that a physical quantity $y$ can be expressed as a function of other physical quantities $x_1,\ldots,x_p$, where $x_1\sim x_d$ are mutually independent, and $x_{d+1}\sim x_p$ are not independent of $x_1\sim x_d$, then
\begin{equation}    y=f(x_1,\ldots,x_p)=f_1(\Pi_1,\ldots,\Pi_{p-d})\prod_{j=1}^d x_j^{q_j},
\label{pi}
\end{equation}
where $\Pi_i=\frac{x_{d+i}}{\prod_{j=1}^d x_j^{q_{ij}}} (i=1,\ldots,p-d)$ are dimensionless numbers.
\end{theorem}

\begin{theorem}
{\rm(Taylor's Theorem\cite{buckingham1914physically})} Let $f(\B x)$ be a function that is $(n+1)$ times differentiable at a point $\hat{\B x}$,  $f(\B x)$ can be approximated by a Taylor polynomial of degree $n$ around $\hat{\B x}$ as:
\begin{equation}
\begin{split}
    f(\B x)&=f(\hat{\B x})+f'(\hat{\B x})(\B x-\hat{\B x})+\frac{f''(\hat{\B x})}{2!}(\B x-\hat{\B x})^2 \\
    &+\ldots+\frac{f^{(n)}(\hat{\B x})}{n!}(\B x-\hat{\B x})^n + R_n(\B x),
\end{split}
\end{equation}
where $R_n(\B x)=\frac{f^{(n+1)}(\boldsymbol\xi)}{(n+1)!}(\B x-\hat{\B x})^{n+1}$ represents the Lagrange remainder, and $\boldsymbol\xi$ lies between $\hat{\B x}$ and $\B x$.
\end{theorem}

\subsection*{Two-Stage Composite Function} 

Given a dataset $(\Omega_x, \Omega_y)$, where $\Omega_x\in \mathbb{R}^{b\times p}$ is the input composed of variables $\B x\in \mathbb{R}^{p}$, $\Omega_y\in \mathbb{R}^{b\times 1}$ is the output composed of $y\in \mathbb{R}$, we aim to discover the underlying natural laws represented by a mapping $y=f(\B x)$. 
In favor of a concise and elegant equation, here we propose to decompose $f(\B x)$ into a two-stage explicit composite function, consisting of a latent layer $\B z=f_1(\B x)\in \mathbb{R}^s$ and an expression layer $y=f_2(\B z)$, with the scheme shown in Fig.~\ref{framework}a.

For the latent layer, drawing inspiration from Buckingham $\Pi$ theorem, we have the following corollary.
\begin{corollary}
Assuming that a physical quantity $y$ can be expressed as a function of other physical quantities $x_1,\ldots,x_p$, involving $d$ fundamental physical quantities, we have $y=f(z_1,\ldots,z_{p-d+1})$ with $z_i=\prod_{j=1}^{p} {x_j^{w_{ij}}}$.
\label{coro}
\end{corollary}
\begin{proof}
In Eq.~(\ref{pi}), both $\Pi_i (i=1,\ldots,p-d)$ and $\prod_{j=1}^dx_j^{q_j}$ take the form $\prod_{j=1}^{p}{x_j^{w_{ij}}}$, with a total of $p-d+1$ terms.
\end{proof}
\noindent According to this corollary, if there are $d$ fundamental physical quantities, the relationship between $y$ and $x_1, \ldots, x_p$ can involve at most $p - d + 1$ latent variables in the form of power products. 
Based on this, we define the latent variables as
\begin{equation}
    z_i=\prod_{j=1}^{p}x_j^{w_{ij}}, i=1,\ldots,s
\label{latent}
\end{equation}
with $\B W = <w_{ij}> \in \mathbb{R}^{s\times p}$ being the power matrix and $s\leq p-d+1$, and the latent layer $f_1$ transforms the input $\B x$ into $\B z$, achieving dimensionality reduction and making meaningful combinations of inputs. 

Regarding the expression layer, Taylor's theorem tells that most functions can be expanded into a Taylor series, allowing us to approximate them using polynomials. Therefore, $f_2$ can be expressed as 
\begin{equation}
    y=a_0+\sum_{i=1}^s a_iz_i+\sum_{i_1,i_2=1}^s a_{i_1i_2}z_{i_1}z_{i_2}+\ldots,
\end{equation}
serving as a unified proxy of the target $f_2(\cdot)$. 
If the resulting expression is overly complex, after identifying the optimal latent variable, we can perform symbolic regression to convert it into a more compact symbolic expression that well reflects the physical meaning of diverse physical formulas.


\subsection*{Structure Identification} 

To reduce the search space, we first analyze the dataset to extract a small number of contributing inputs, estimate the number of latent variables, and identify the input-to-latent connections, as shown in Fig.~\ref{framework}b. 

\vspace{2mm}
\noindent\textbf{Identification of Contributing Inputs.\quad} In this study, we employ SHAP (Shapley Additive Explanations)\cite{lundberg2017unified} to quantify the contribution of each input variable to the output. By ranking these contributions, we can pinpoint the key variables and discard any that are irrelevant.
To overcome the limitation of the Buckingham $\Pi$ formulation, which only allows for identifying latent variables in a multiplicative power-product form (like $z_1=x_1^ax_2^b$), we expand the input space by generating new candidates through unary operations $\{\text{cos}, \text{sin}, \text{tan}, \text{exp}, \text{abs}, \ldots\}$ and binary operations $\{+,-,\times,/\}$. We then apply SHAP to this enriched feature set to highlight the most impactful combinations as refined inputs. This approach helps capture the nonlinear interactions beyond the multiplicative power-product form.

\vspace{2mm}
\noindent\textbf{Identification of Topological Connections between Input and Latent Variables.\quad} 
From Eq.~(\ref{latent0}), we can get
\begin{equation}
\begin{split}
    &\f{\p z_i}{\p x_j}=w_{ij}x_j^{w_{ij}-1}\prod_{k\neq j}x_k^{w_{ik}}
    =w_{ij}\f{z_i}{x_j}, \\
    &\f{\p y}{\p x_j}=\sum_{i=1}^{s}\f{\p z_i}{\p x_j}\f{\p y}{\p z_i}
    =\sum_{i=1}^{s}w_{ij}\f{z_i}{x_j}\f{\p y}{\p z_i},
\end{split}
\label{deri}
\end{equation}
and if $x_j$ and $x_k$ only connect to $z_i$, we have
\begin{equation}
    \f{\p y}{\p x_j}/\f{\p y}{\p x_k}=(w_{ij}\f{z_i}{x_j}\f{\p y}{\p z_i})/
    (w_{ik}\f{z_i}{x_k}\f{\p y}{\p z_i})= \f{w_{ij}x_k}{w_{ik}x_j}.
\end{equation}
We further define $\rho_j=x_j\f{\p y}{\p x_j}$ as an attribute of $x_j$, leading us to conclude that the ratio $\f{\rho_j}{\rho_k}$ is a constant, i.e.,
\begin{equation}
    \f{\rho_j}{\rho_k}=\f{x_j\f{\p y}{\p x_j}}{x_k\f{\p y}{\p x_k}}=\f{w_{ij}}{w_{ik}}.
\label{rho}
\end{equation}
Based on the above equations, when $\rho_j$ and $\rho_k$ at different data points exhibit a linear relationship $\rho_j=c\rho_k$, we can infer that $x_j$ and $x_k$ are likely connected to the same latent variable $z_i$ and $w_{ij}/w_{ik}\approx c$. 

To estimate the $\boldsymbol\rho$ value, we can either use the backward difference method or apply first-order differentiation to the polynomial fitting the adjacent points. 
From the calculated $\boldsymbol\rho$, we use the PPMCC to evaluate the linear correlation between $\rho_j$ to $\rho_k$ ($j,k\in[1,p]$) and yield a $p\times p$ correlation matrix, from which we can easily identify the relevant inputs, the number of latent variables, and the connection relationships between them. 
Finally, we employ the least squares method to estimate the slope between $\rho_j$ and $\rho_k$, which allows us to estimate the value of $\f{w_{ij}}{w_{ik}}$.




\subsection*{Parameter Constraints}

\vspace{2mm}
\noindent\textbf{Dimensional Invariance.\quad}
To ensure that the resulting equation has explicit physical meaning, it is necessary to ensure $f_2\circ f_1(\B x)$ and $y$ have consistent unit, i.e.,
\begin{equation}
    f_2(\B D\B W^T)=\B d,
\end{equation}
where $\B D\in \mathbb{R}^{d\times p}$ is the dimension matrix of $\B x$ and $\B d\in \mathbb{R}^{d}$ is the dimension matrix of $y$. As shown in Fig.~\ref{framework}c, 
the dimension matrix represents the exponents of physical quantities concerning the fundamental dimensions in the natural world---mass [M], length [L], time [T], temperature [$\Theta$], electric current [I], luminous intensity [J], and amount of substance [N]. 

From Eq.~(\ref{latent}), the $i$-th latent variable $z_i$ is determined by the $i$-th row of $\B W$, i.e., $\B W_i$. When the output $y$ has no unit, i.e., $\B d=\B 0$, we set $\B D\B W_i^T=\B 0$ to ensure $z_1,\ldots,z_s$ are dimensionless numbers. When $y$ has a unit, i.e., $\B d\neq\B 0$, we set $\B D\B W_i^T= \B d$ to ensure unit consistency between the latent variables and the output. Incorporating the above two cases, we have the following constraints on the weight
\begin{equation}
    \B D\B W_i^T=\B d,
\end{equation}
whose closed-form solution is 
\begin{equation}
    \B W_i^T=\sum_{k=1}^{p-r(\B D)}\boldsymbol\lambda_{ik}\B E_{k}+\B e^*=\B E\boldsymbol\lambda_{i}^T+\B e^*,
    \label{lambda}
\end{equation}
where $\{\B E_k\}$ is the set of homogeneous solutions that satisfies $\B D \B E_k=0$, $\B E=[\B E_1,\ldots,\B E_{p-r(D)}]$, and $\B e^*$ is a particular solution to  $\B D \B e^*=\B d$. Once we have obtained $\B E$ and $\B e^*$, the task of searching for $\B W\in\mathbb{R}^{s\times p}$ turns into searching for $\boldsymbol\lambda\in\mathbb{R}^{s\times (p-r(\B D))}$, with $r(\cdot)$ denoting the rank of a matrix.

In most cases, each latent variable is only a partial combination of input variables. After structure identification, we can determine the inputs related to $z_i$, i.e., we know the indices of zero and non-zero entries in $\B W_i$ and denote them as $\xi_i$ and $\tau_i$, respectively:
\begin{equation}
\left\{
\begin{array}{cc}
    \B W_i^T[\xi_i]=\B E[\xi_i]\boldsymbol\lambda_{i}^T+\B e^*[\xi_i]=0, & \\
    \B W_i^T[\tau_i]=\B E[\tau_i]\boldsymbol\lambda_{i}^T+\B e^*[\tau_i]\neq0. & 
\end{array}
\right.
\end{equation}
The closed-form solution of this equation is
\begin{equation}
\begin{array}{ll}
   &\boldsymbol\lambda_{i}^T=\B F_i\boldsymbol\mu_{i}^T+\B f_i^*, \\
   &s.t.\quad \B E[\tau_i]\B F_i\boldsymbol\mu_{i}^T+\B E[\tau_i]\B f_i^*+\B e^*[\tau_i]\neq0. 
\end{array}
\end{equation}
where $\B F_i\in\mathbb{R}^{(p-r(\B D))\times (p-r(\B D)-r(\B E[\xi_i]))}$ is the homogeneous solutions that satisfies $\B E[\xi_i] \B F_i=0$, and $\B f_i^*$ is a particular solution to  $\B E[\xi_i] \B f_i^*=-\B e^*[\xi_i]$. Once we have obtained $\B F_i$ and $\B f_i^*$, the task of searching for $\boldsymbol\lambda$ turns into searching for $\{\boldsymbol\mu_i\}$.

Here we propose two strategies for imposing the DI constraints: DI-1 incorporates dimensional constraints during the search process, enforcing consistent units between latent variables and the output. This approach significantly reduces the search space. DI-2 performs dimension-agnostic searching by setting $\B D=\B 0, \B d=\B 0$. In this case, $\B E=\B I, \B e^*=\B 0, \B W=\boldsymbol\lambda$. Since the units of latent variables and the output in the resulting equations might differ, we subsequently assign appropriate units to the learned constant coefficients (i.e., $\{a_i\}$ in $f_2$) to compensate the inconsistency.
Both approaches have their respective advantages and disadvantages. DI-1 features a smaller search space, faster computation, and strict dimensional consistency between the latent variables and the output, providing better physical interpretability. 
In contrast, DI-2 has a larger search space but offers greater flexibility by allowing assigning units to the constant coefficients in $f_2$. For example, in the universal gravitation formula $F=\frac{G m_1 m_2}{r^2}$, $G$ is a constant with units, which can be discovered through the DI-2 framework.

\vspace{2mm}
\noindent\textbf{Dataset Constraints.\quad} The form of the latent variables provides valuable constraints on the function parameters, which can help narrow down the search space even further.
From Eq.~(\ref{latent}), there exists $x_j^{w_{ij}}$ term in $z_i$, so (i) if $\exists x_j<0$ in the dataset $(\Omega_x, \Omega_y)$, we let $w_{ij}\in\mathbb{Z}$ to avoid the occurrence of imaginary numbers; (ii) if $\exists x_j=0$ in the dataset, we force $w_{ij}\geq 0$ to avoid a zero divisor.
\begin{equation}
    \left\{
    \begin{array}{cc}
        w_{ij}\in\mathbb{Z}, & \text{if } \exists x_j<0\\
        w_{ij}\geq 0. & \text{if } \exists x_j=0
    \end{array}
    \right.
\end{equation}

\vspace{2mm}
\noindent\textbf{Equivalence Constraints.\quad} 
We can eliminate searching for coefficients equivalent to those that have already been optimized, which happens a lot when multiple latent variables are involved. For instance, we avoid using a repetitive latent variable, which makes optimizing the corresponding coefficients unnecessary. Additionally, the order of the latent variables does not impact the outcome, e.g., the expressions  $z_1=x_1^{1}x_2^{2}x_3^{3}, z_2=x_1^{4}x_2^{5}x_3^{6}$ is equivalent to $z_1=x_1^{4}x_2^{5}x_3^{6}, z_2=x_1^{1}x_2^{2}x_3^{3}$, so the optimization of any coefficient matrices generated by swapping two rows of a previously optimized matrix \( \B W \) are unnecessary.

\vspace{2mm}
\noindent\textbf{Sparsity Constraints.\quad} In most cases, each latent variable is only composed of a partial combination of input variables, i.e., $\B W$ is a sparse matrix represented as
\begin{equation}
    w_{ij}\left\{
    \begin{array}{cc}
    =0, & x_j\nrightarrow z_i\\
    \neq 0. & x_j\rightarrow z_i.
    \end{array}
    \right.
\end{equation}
To find concise and meaningful input combinations, we impose restrictions on the sparsity of the data. Specifically, we force $\B W$ to have at most $\kappa_1$ non-zero values and each column has no more than $\kappa_2$ non-zero entries, i.e., each input is associated with at most $\kappa_2$ latent variables:
\begin{equation}
\begin{split}
\left\{
\begin{array}{l}
    \Vert\{w_{ij}|w_{ij}\neq 0,i=1,\ldots,s,j=1,\ldots,p\}\Vert\leq\kappa_1, \\
    \Vert\{w_{ij}|w_{ij}\neq 0,i=1,\ldots,s\}\Vert\leq\kappa_2,~j=1,\ldots,p.
\end{array}
\right.
\end{split}
\end{equation}
It is important to highlight that the constraint is optional. You can choose to apply it when you have enough prior knowledge about $\kappa_1$ and $\kappa_2$.

\vspace{2mm}
\noindent\textbf{Sign Encoding.\quad} 
In the C2F process, the newly acquired coefficients from each refinement might have already been explored. To streamline our searches and avoid any unnecessary duplication, we assign codes to these coefficients based on their signs and leverage Python's dictionary capabilities to group the coefficients sharing the same sign code into distinct categories. This way, when we look for new coefficients, we only need to check them against those in their respective sign code groups to see if they've been searched before.

\vspace{2mm}
\noindent\textbf{Integration of the Constraints.\quad}
The five strategies outlined above are quite straightforward, yet they pack a powerful punch in terms of efficiency. When dealing with large-scale input or latent variables, these approaches can shrink the search space from hundreds of millions down to just a few dozen, which significantly reduces resource consumption. To ensure we do not overlook the true solution, we typically refrain from using the sparse constraint. However, all other strategies are activated by default and do not hinder our pursuit of the optimal solution.

\subsection*{Parameter Optimization}

\vspace{2mm}
\noindent\textbf{C2F Grid Search.\quad} 
Given an estimated version of $\boldsymbol\mu$, denoted as $\hat{\boldsymbol\mu}$, the mapping towards the latent variable can be estimated as
\begin{equation}    \hat{\Omega_z}=f_1(\Omega_x|\hat{\boldsymbol\mu}).
\label{C2F1}
\end{equation}
Then we minimize the least squares error to perform polynomial regression on $\hat{\Omega_z}$ and $\Omega_y$, and obtain $f_2$’s estimation $f_2(\cdot|\hat{\boldsymbol\mu})$ , i.e., 
\begin{equation}
    f_2(\cdot|\hat{\boldsymbol\mu})=\underset{f_2}{\arg\min}\lVert\Omega_y-f_2(\hat{\Omega_z})\rVert_2.
\label{C2F2_f2}
\end{equation}
In other words, the polynomial coefficients of $f_2$ can be calculated easily from $\hat{\boldsymbol\mu}$, and the optimization of the coefficients in the target function turns into the search for the optimal $\hat{\boldsymbol\mu}$.

Since the output $\hat{\Omega_y}$ can be predicted from the given $\hat{\boldsymbol\mu}$ as 
\begin{equation}
\hat{\Omega_y}=f_2(\hat{\Omega_z}|\hat{\boldsymbol\mu}),
\label{C2F2_output}
\end{equation}
we use the coefficient of determination
\begin{equation}
    R^2=1-\f{\sum_{i=1}^{b}({\Omega_y}_i-\hat{{\Omega_y}_i})^2}{\sum_{i=1}^{b}({\Omega_y}_i-\bar{\Omega_y})^2}
\label{C2F3}
\end{equation}
to measure the data fitting accuracy, with $\bar{\Omega_y}=(\sum_{i=1}^{b}{\Omega_y}_i)/b$ being the mean of $\Omega_y$. Naturally, $\boldsymbol\mu$ can be estimated by the following optimization
\begin{equation}
    \dot{\boldsymbol\mu}=\underset{\hat{\boldsymbol\mu}}{\arg\min}~R^2(\Omega_y,f_2(\hat{\Omega_z}|\hat{\boldsymbol\mu})).
\label{C2F4}
\end{equation}
 
In most cases, people prefer latent variables generated by input combinations with small exponents, such as the law of universal gravitation \( F = G m_1 m_2 r^{-2} \) and Kepler's third law \( T = k r^{1.5} \), so we limit \( \boldsymbol{\mu}_i \) to the range \([-2,2]^{p - r(\mathbf{D}) - r(\mathbf{E}[\xi_i])}\), and there are \( c = \sum_{i=1}^{s} \left( p - r(\mathbf{D}) - r(\mathbf{E}[\xi_i]) \right) \) parameters to be searched.
There are several typical searching strategies: gradient optimization algorithms tend to get stuck in local optima, which often leads to irregular decimal values such as 0.5234 instead of more concise ones like 0.5. On the other hand, linear searching can avoid the problem of local optima but faces challenges due to the large search space. 
To address these issues, we propose a C2F (Coarse-to-Fine) optimization framework that searches for the optimal solution in a systematic manner. As illustrated in Fig. \ref{framework}c, we begin with a rough initial value and progressively refine the solution within increasingly smaller ranges, moving closer to the optimal value.
\begin{itemize}[leftmargin=*]
    \item  \em{Initialization}. We first divide $[-2,2]^c$ with a step of 1.0 to obtain $5^c$ initial estimations of $\boldsymbol\mu$ and then exclude a lot of infeasible searching candidates with the constraints proposed before. For the rest estimations, one can select the top candidates based on their corresponding $R^2$ scores calculated from Eqns.~(\ref{C2F1},\ref{C2F2_f2},\ref{C2F2_output},\ref{C2F3},\ref{C2F4}). 
    \item  \em{Refinement}. We perform the next round of search with a step of 0.5 around the selected top candidates, calculate the $R^2$ distribution, and update them successively. Then we repeat this process progressively, with the step decreasing from 0.5 to 0.2 or even to 0.1.
\end{itemize}

The C2F search minimizes the risk of local optima and significantly narrows the search space. This strategy also offers flexibility regarding initialization settings and step size.
Moreover, our optimization strategy tends to overlook exponents of irregular decimals, potentially leading to results that better reflect the intrinsic simplicity of scientific laws.

\vspace{2mm}
\noindent\textbf{Multi-Level Optimization.\quad} 
The weight matrix $\B W$ contains $s \times p$ values. The initialization step in the above C2F grid search requires \( 5^{s \times p} \) searches, and this number will become prohibitively large when there are many latent variables.
To tackle this issue, we propose a progressive strategy to optimize the latent variables, as illustrated in Fig.~\ref{framework}c. Specifically, we begin by optimizing one variable, and then fix it and move on to optimize a new one. This process is repeated progressively to enhance the performance metric until we meet the predefined criteria.
This progressive optimization reduces the initialization search space from $5^{s \times p}$ to $s \times 5^p$, while experimentally maintaining comparable performance, as detailed in the Supplementary Material.

\subsection*{Expression Simplification}
As illustrated in Fig.~\ref{framework}d, upon determining the optimal latent variable and its corresponding polynomial expression, we proceed to further simplify the expression if the polynomial is excessively intricate. Taking the solution $y=sin(\f{x_1^2x_3^{1.5}}{\sqrt{x_2}})$ as an example, following the application of the structure identification and parameter optimization procedures, we derive $z=\f{x_1^2x_3^{1.5}}{\sqrt{x_2}}$ and $y=z-\f{z^3}{3!}+\f{z^5}{5!}$. To achieve a more compact representation, we utilize PySR\cite{cranmer2023interpretable} to conduct symbolic regression on the data points $\Omega_z$ and $\Omega_y$, which transforms the expression into a more succinct form $y = sin(z)$. 

Applying symbolic regression directly to Eq.~(\ref{C2F2_f2}) would lead to extremely long computation times for each test of $\hat{\boldsymbol\mu}$. In contrast, our approach significantly accelerates the process by starting with polynomial regression to achieve a temporary substitute for the original function efficiently and then applying symbolic regression to the resulting polynomial.

\vspace{4mm}
\noindent\textbf{Data Availability\quad}
All datasets used in this study are available on GitHub at \href{ https://github.com/HarryPotterXTX/FIND}{https://github.com/HarryPotterXTX/FIND}.

\vspace{2mm}
\noindent\textbf{Code Availability\quad}
All source codes used in this manuscript are available on GitHub at \href{ https://github.com/HarryPotterXTX/FIND}{https://github.com/HarryPotterXTX/FIND}.


\bibliography{ref}


\section*{Acknowledgements}
This work is jointly funded by National Natural Science Foundation of China (Grant No. 61931012) and the Beijing Municipal Natural Science Foundation (Grant No. L257009).

\section*{Author Contributions}
J. S., and T. X. conceived this project. J. S. supervised this research. T. X. designed the FIND architecture. T. X., X. S., Z. W., and B. Z. conducted the experiments and data analysis. All the authors participated in the writing of this paper.

\section*{Competing Interests}
The authors declare no competing financial interests.

\section*{Materials and Correspondence}
Correspondence and material requests should be addressed to Jinli Suo.

\end{multicols}

\end{document}